\documentclass[12pt,onecolumn, draftcls]{IEEEtran}
\usepackage{graphicx}
\usepackage{amssymb, amsmath, graphicx, cite, color, epsfig, subfigure, mathdots}

\newtheorem{theorem}{Theorem}

\newtheorem{definition}{Definition}

\newtheorem{remark}[definition]{Remark}

\newcommand{\bigO}[1]{O{\left(#1\right)}}
\newcommand{\smallO}[1]{o{\left(#1\right)}}

\begin{document}

\title{The Degrees of Freedom of the $K$-pair-user Full-Duplex Two-way Interference Channel with and without a MIMO Relay}
\author{\IEEEauthorblockN{Zhiyu Cheng, Natasha Devroye\\
University of Illinois at Chicago \\
zcheng3, devroye@uic.edu
}
}
\maketitle

\begin{abstract}

In a $K$-pair-user two-way interference channel (TWIC), $2K$ messages and $2K$ transmitters/receivers form a $K$-user IC  in the forward direction ($K$ messages) and another $K$-user IC in the backward direction which operate in full-duplex mode. All nodes may  interact, or adapt inputs to past received signals. We derive a new outer bound to demonstrate that the optimal degrees of freedom (DoF, also known as the multiplexing gain) is $K$: full-duplex operation doubles the DoF, but interaction does not further increase the DoF. We next characterize the DoF of the $K$-pair-user TWIC with a MIMO, full-duplex relay.  If the relay is non-causal/instantaneous (at time $k$ forwards a function of its received signals up to time $k$) and has $2K$ antennas, we demonstrate a one-shot scheme where the relay  mitigates all interference to achieve the interference-free $2K$ DoF.  In contrast, if the relay is causal (at time $k$ forwards a function of its received signals up to time $k-1$), we show that a full-duplex MIMO relay cannot increase the DoF of the $K$-pair-user TWIC beyond $K$, as if no relay or interaction is present. 
We  comment on reducing the number of antennas at the instantaneous relay.

\end{abstract}

\section{Introduction}


In wireless communications, current two-way systems often  employ either time or frequency division to achieve two-way or bidirectional communication. This restriction is due to a combination of hardware and implementation imperfections and effectively orthogonalizes the two directions, rendering the bidirectional channel equivalent to two one-way communication systems. 
However, much progress has been made on the design of full-duplex  wireless systems \cite{Katti:2011, Sahai:2012}, which show great promise for increasing data rates in future wireless technologies. In this work we seek to understand the potential of full duplex systems in a two-way multi-user or network setting, and do so from a multi-user information theoretic perspective by obtaining the degrees of freedom of several full-duplex two-way networks with and without relays. 

Full-duplex operation enables true two-way communications over the practically relevant Gaussian noise channels. We currently understand the theoretical limits of a point-to-point, full-duplex Gaussian two-way channel where two users wish to exchange messages over in-band Gaussian channels in each direction: the capacity region is equal to two independent Gaussian noise channels operating in parallel \cite{Han:1984}. Full duplex operation thus roughly doubles the capacity of this simple two-way network. 

To extend our understanding of the impact of full-duplex operation to {\it two-way networks with interference}, the two-way interference channel  (TWIC) has been studied in \cite{zcheng_ISIT, Cheng:2012, zcheng_Allerton2012, Suh2012}, in which there are 4 independent messages:  two-messages to be transmitted over an interference channel (IC) in the $\rightarrow$ direction simultaneously with two-messages to be transmitted over an in-band IC in the $\leftarrow$ direction. 
All 4 nodes in this network act as both sources and destinations of messages. This allows for {\it interaction} between the nodes: a node's channel inputs may be functions of its message and previously received signals. The capacity region of the point-to-point  two-way channel is still open in general, though we know the capacity for the Gaussian channel, and is known to be a remarkably difficult problem. Similarly, the capacity region of the one-way IC is still open, though we know its capacity to within a constant gap for the Gaussian noise channel \cite{etkin_tse_wang}. In general then, finding the full capacity region of the full-duplex  TWIC is a difficult task, though progress has been made for several classes of deterministic channel models \cite{Cheng:2012}, and capacity is known to within a constant gap  in certain parameter regimes and adaptation constraints \cite{zcheng_Allerton2012,Cheng:2012, Cheng:GlobalSIP}. 

The degrees of freedom (DoF) \cite{Zheng: Dof} is an alternative (to the all out, challenging capacity region) approximate capacity characterization that intuitively corresponds to the number of independent interference-free signals that can be communicated in a network at high signal to noise ratios (SNR) which has been of significant recent interest in one-way  interference networks \cite{Jafar:2008:alignment, Jafar:2009:relays}. Here, we seek to extend our understanding of the DoF to  two-way networks, whose study is motivated by the fact that full-duplex operation is becoming practically realizable. 
Some progress has already been made: the DoF of the full-duplex  TWIC has been shown to be 2 \cite{zcheng_ISIT, zcheng_Allerton2012}. This is interesting, because  the capacity of any network with interaction at nodes is no smaller than that of the same network where interaction is not possible (interaction can mimic non-interaction).  However, that the TWIC with interact has DoF 2  demonstrates that  interaction between users does {\it not} increase the DoF of the two-way IC beyond the doubling that full-duplex operation gives. We ask whether the same is true for $K$-pair user two-way, full-duplex interference channels with and without a MIMO relay node.



\subsection{Contributions}
In this work, we first propose and study a natural extension of the (2-pair-user) two-way interference channel (TWIC): the $K$-pair-user two-way interference channel, i.e., there are $2K$ messages and $2K$ users forming a $K$-user IC ($K$ messages) in the forward direction and another $K$-user IC in the backward direction ($K$ messages). We consider this $2K$ node network with and without the presence of a MIMO relay. 
All nodes may employ interaction -- i.e. signals may be a function of previously received outputs. Compared to the 2-pair-user IC,  the $K$-pair-user two-way IC experiences interference from significantly more users:  not only may a receiver see signals from all other transmitters transmitting in the same direction, but due to the adaptation involved, these signals may also contain information about users transmitting in the opposite direction. Hence, any user may see a combination of the signals of all other $2K-1$ users in addition to seeing 
self-interference (SI) signals, which are transmitted  by the user itself or received via other signals due to adaptation.  
Canceling SI is one of the main challenges in real full-duplex wireless systems.  However, in this theoretical work for the Gaussian channels involved,   the self-interference is known to the receiver and as such, theoretically, it can be subtracted off. We then explore the limits of communication under the assumption that this self-inteference may be removed. Our main results are:

1) We first show that the sum degrees of freedom of the $K$-pair-user TWIC  is $K$, i.e. $K/2$ in each direction, for both time-varying and (almost all) constant channel coefficients. In other words, each user still gets half a DoF and interaction between users -- even though our outer bound permits it -- is again useless from a DoF perspective. Intuitively this is because all the links in the network have similar strengths in the DoF sense, so that a user cannot ``route'' other users' desired signals through backward links since they are occupied by its own data signals. 
 In addition, coherent power gains which may be the result of adaptation and the ability of nodes to correlate their channel inputs, do not affect the DoF (i.e. coherent power gains for Gaussian channels lead to additive power gains inside the logarithm rather pre-log, or DoF/multiplexing gains). Achievability follows from known results of the one-way $K$-user IC; the contribution lies in the novel outer bound.
 Full-duplex operation is thus seen to double the DoF, but interaction is not able to increase the DoF beyond this.

2)  We next consider the $K$-pair-user TWIC with an additional, multi-antenna, full-duplex relay node which does not have a message of its own and only seeks to aid the communication of the $K$-pair users. We ask whether the presence of such a relay node may increase the DoF. 
Interestingly, we show that while the DoF of the $K$-pair-user  TWIC is $K$ -- indicating that interference is present and somewhat limiting rates in the $K$-user IC in each direction -- that the presence of an {\it instantaneous} MIMO relay with $2K$ antennas may increase the DoF to the maximal value of $2K$, i.e. each user in the network is able to communicate with its desired user in a completely interference-free (in the DoF sense) environment. The key assumption needed is for the relay to be  non-causal or  instantaneous -- meaning that at time $k$ it may forward a signal based on the received up to and including time $k$. 
We see that full-duplex operation combined with instantaneous / non-causal relaying with multiple antennas may in this case quadruple the DoF over the one-way $K$-user IC. 

3) Finally, we show a result which is sharp contrast to the previous point: if the relay is now causal instead of non-causal, meaning that at time $k$ it may only forward a signal which depends on the received signals up to and including time $k-1$, then we derive a novel outer bound which shows that the DoF of the $K$-pair-user TWIC with a causal MIMO relay is $K$ (regardless of the number of antennas at the relay). This is the same as that achieved without a relay, and without interaction. In summary,  full-duplex operation again doubles the DoF, but a causal, full duplex relay is unable to increase the DoF beyond that. 
 

%
 
\subsection{Related Work}

The degrees of freedom of a variety of one-way communication networks have been characterized \cite{Jafar:2009:relays, Jafar:2009:X, Etkin:dof, Ke2011, Vaze:dof}. However, much less is known about the DoF of two-way communications. Very recently, \cite{Lee:2013} considered a {\it half-duplex} two-pair two-way interference channel (where nodes other than the relay may {\it not} employ interaction and hence are much more restricted than the nodes here, i.e. transmit signals are functions of the messages only and not past outputs) with a 2-antenna relay and showed that 4/3 DoF are achievable.  No converse results where provided.  In \cite{Wang:2013}, the authors identified the DoF of the full-duplex 2-pair and 3-pair two-way multi-antenna relay MIMO interference channel, in which there is no interference between users who only communicate through the relay (no direct links). 
We consider direct links between all users (not in the same sides) in the two-way interference channels, as well as links between all users and the relay. We also note that the general results of \cite{Jafar:2009:relays}, which state that relays, noisy cooperation, perfect feedback and full-duplex operation does not increase the DoF of one-way networks,  do not apply, as we consider nodes which are both sources and destinations of messages (two-way rather than one-way).

The $K$-user interference channel, as an extension of the 2-user interference channel, information theoretically models wireless communications in networks involving more than two-pairs of users. Using the idea of interference alignment \cite{maddah-ali:IA, Jafar:2008:alignment, Jafar:IA}, the DoF of the $K$-user (one-way) IC for both time-varying channels and (almost all)\footnote{The precise definition of ``almost all'' may be found in \cite{Motahari2009}.} constant channels has been shown to be $K/2$ in \cite{Jafar:2008:alignment} and \cite{Motahari2009} respectively. The generalized DoF of the $K$-user IC without and with feedback have been characterized in \cite{Jafar_Gdof2010} and \cite{Mohajer2012} (full feedback from receiver $i$ to transmitter $i$) respectively. Authors in \cite{Shomorony_Allerton} showed that for almost all constant channel coefficients of  fully connected two-hop wireless networks with $K$ sources, $K$ relays and $K$ destinations (source nodes are not destination nodes as they are here, i.e. the network is one-way), the DoF is $K$.

We note that our work differs from prior work in that we consider an {\it interactive}, {\it full-duplex} Gaussian $K$-pair-user TWIC for the first time, with and without a relay (which may be either non-causal or causal), and obtain not only sum-rate achievability but also converse DoF results for all three general channel models considered. We emphasize that we seek information theoretic DoF results, which act as benchmarks / upper bound for the performance of realistic scenarios. 


\subsection{Outline}
We present the system model for the $K$-pair-user TWIC with and without a relay in Section \ref{model}. 
We derive a new outer bound  to show that $K$ DoF is optimal for the Gaussian $K$-pair-user TWIC in Section \ref{KDoF}; achievability follows by considering two non-adaptive one-way $K$-user IC schemes. 
Then we proceed to consider the $K$-pair-user TWIC with an instantaneous MIMO relay in Section \ref{2KDoF}, where we show that the maximum $2K$ DoF  may be achieved with the help of an instantaneous MIMO relay with at least $2K$ antennas. We demonstrate a one-shot achievability scheme. We comment on the possibility of reducing the number of antennas at the instantaneous relay node. In Section \ref{sec:causal} we then show that if the relay is causal rather than non-causal, that, for the $K$-pair-user TWIC, the presence of a full-duplex, multi-antenna relay cannot increase the DoF beyond $K$ (which is achievable without relays and without interaction, but requires full-duplex operation). This is done by developing a new outer bound which allows for interaction and causal relaying.  We conclude the paper in Section \ref{conclusion}.

\section{System Model}
\label{model}
We describe the $K$-pair-user TWIC without and with a relay in this section.
\subsection{$K$-pair-user two-way interference channel}
We consider a $K$-pair-user TWIC as shown in Fig. \ref{fig:kuseric}, where  there are $2K$ messages and $2K$ terminals forming a $K$-user IC in the $\rightarrow$ direction ($K$ messages) and another $K$-user IC in the $\leftarrow$ direction ($K$ messages). All nodes are able to operate in full-duplex mode, i.e. they can transmit and receive signals simultaneously. 
\begin{figure}
\begin{center}
\includegraphics[width=2in]{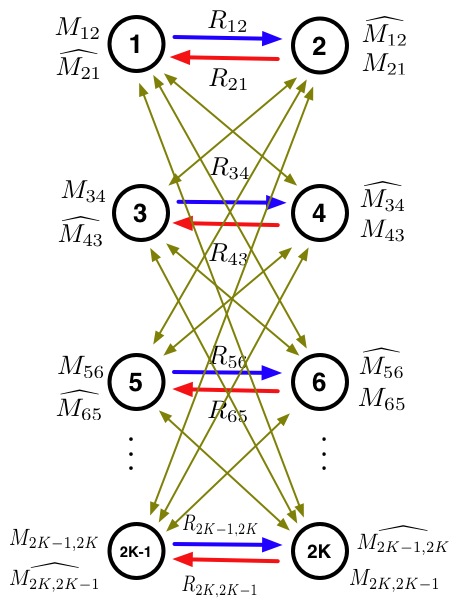}
\caption{$K$-pair-user two-way interference channel. $M_{ij}$ denotes the message known at node $i$ and desired at node $j$ of rate $R_{ij}$; $\widehat{M_{ij}}$ denotes that $j$ would like to decode the message $M_{ij}$ from node $i$.} 
\label{fig:kuseric}
\end{center}
\end{figure}

At each time slot $k$, the system input/output relationships are described as:
\begin{align}
&Y_p[k]=\sum_{m=1}^Kh_{2m,p}[k]X_{2m}[k]+Z_p[k], \ \ \ p=1,3,...,2K-1\label{sm1}\\
&Y_q[k]=\sum_{m=1}^Kh_{2m-1,q}[k]X_{2m-1}[k]+Z_q[k], \ \ \ q=2,4,...,2K\label{sm2}
\end{align}
where $X_l[k], Y_l[k], l \in \{1,2,...,2K\}$ are the inputs and outputs of user $l$ at time slot $k$, and $h_{ij}[k], i,j \in \{1,2,...,2K\}$ is the channel coefficient from node $i$ to node $j$ at time slot $k$.\footnote{Note that If user $l$ is transmitting, we have already assumed that its own ``self-interference'' signal has been ideally subtracted off its received signal, and is hence not present in the above description of channel inputs and outputs. This idealization will of course form an upper bound on what is possible if full self-interference cancellation is not possible, which is outside the scope of this paper and is an interesting topic for future work.}
The network is subject to complex Gaussian noise $Z_l[k] \sim \mathcal{CN} (0,1), l\in \{1,2,...,2K\}$ which are independent across users and time slots. 
We consider time-varying channel coefficients, which for each channel use are all drawn from a continuous distribution (which need not be the same for all channel gains and time instances, as long as they are continuous) and  whose absolute values are bounded between a nonzero minimum value and a finite maximum value. Note one can also alternatively consider a frequency selective rather than time-varying system model.

We further assume per user, per symbol power constraints $E[|X_i[k]|^2]\leq P, i\in \{1,2,...,2K\}, \; k\in\{1,2,\cdots n\}$, for block length $n$.\footnote{In our outer bound in Theorem \ref{kdof}, several of the terms may be extended to per user average power constraints, but we leave the per symbol power constraints for simplicity in this initial study, as is often done in degree of freedom results.}  
User $2i-1$ and $2i$ wish to exchange messages for $i=1,2,\cdots K$ (user 1 sends to 2, 2 to 1,..., 2K-1 to 2K, 2K to 2K-1) 
with {\it interactive}  encoding functions \[ X_i[k]=f(M_{ij}, Y_i^{k-1}), \;\;\;\; k=1,2,\cdots n \] at rate $R_{i,j}=\frac{\log_2|M_{ij}|}{n}$,
where $Y_i^{k-1}$ denotes the vector $(Y_i[1], \cdots Y_i[k-1])$ from time slot, or channel use $1$ to $k-1$ received at user $i$, and $n$ denotes the total number of channel uses (the blocklength). 
In other words,  all users in this network can adapt current channel inputs to previously received channel outputs. The nodes  $2i-1$ and $2i$ have decoding functions which map $(Y_{2i-1}^{n}, M_{2i-1, \; 2i})$ to an estimate of $M_{2i, \; 2i-1}$ and $(Y_{2i}^{n}, M_{2i, \; 2i-1})$ to 
$M_{2i-1, \; 2i}$, respectively. 
A rate tuple $(R_{i,i+1}(P),R_{i+1,i}(P))_{i\in \{1,3,...,2K-1\}}$, where we use the argument $P$ simply to remind the reader that this rate is indeed a function of the  power constraint $P$,  is said to be achievable if there exist a set of interactive encoders and decoders such that the desired messages may be estimated with arbitrarily small probability of error when the number of channel uses $n$ tends to infinity. The sum DoF characterizes the sum capacity of this Gaussian channel at high SNR and is defined as the maximum over all achievable $(R_{i,i+1}(P),R_{i+1,i}(P))_{i\in \{1,3,...,2K-1\}}$ of
\begin{align*}
d_{sum}&=\sum_{i=1,3,...,2K-1}(d_{i,i+1}+d_{i+1,i})\\
&=\limsup_{P\rightarrow\infty} \frac{\sum_{i=1,3,...,2K-1}(R_{i,i+1}(P)+R_{i+1,i}(P))}{\log (P)}.
\end{align*}
Notice the implicit definitions of the DoF of the link from user $i$ to user $i+1$, $d_{i,i+1}$ and the reverse $d_{i+1,i}$. 

\subsection{$K$-pair-user two-way interference channel with a MIMO relay}
\label{modelrelay}
We consider a $K$-pair-user two-way interference channel with a MIMO relay as shown in Fig. \ref{fig:kuser}. All the system settings are the same as in the previous section except there is a MIMO relay which helps in communicating messages and managing interference in the network. As before, all nodes including the relay are able to operate in full-duplex mode, or transmit and receive at the same time over the same channel, and perfectly cancel out their self-interference.

\begin{figure}
\begin{center}
\includegraphics[width=2in]{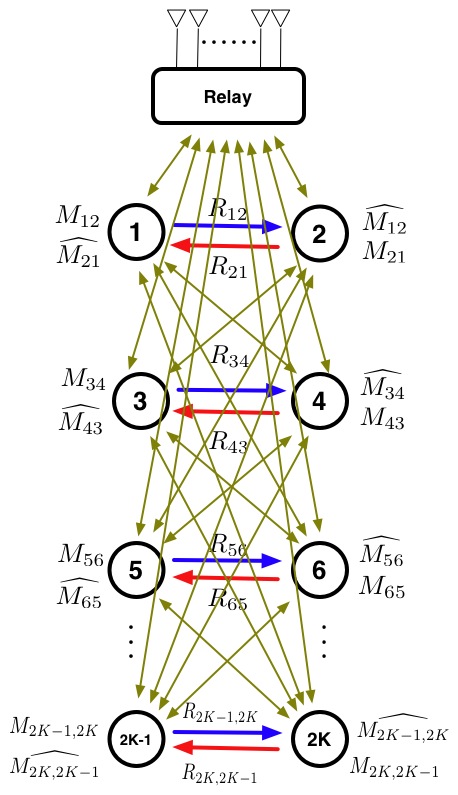}
\caption{$K$-pair-user two-way interference channel with a MIMO relay. $M_{ij}$ denotes the message known at node $i$ and desired at node $j$; $\widehat{M_{ij}}$ denotes that $j$ would like to decode the message $M_{ij}$ from node $i$.} 
\label{fig:kuser}
\end{center}
\end{figure}

The relay is assumed to have $M$ antennas and to operate either in a non-causal or ``instantaneous'' fashion, or in a causal fashion. By ``instantaneous'' (non causal, relay-without-delay \cite{ElGamal2005}) we refer to its ability to decode and forward signals received at the previous and {\it current} (but {\it not} future) time slots. We note that this requirement is significantly less strict than a {\it cognitive} relay, which would know all users' signals prior to transmission and does not obtain the messages over the air. We will comment more on the usage / impact of a cognitive relay in Section \ref{subsec:antennas}. 
Here messages are obtained over the air; the only idealization is the non causality or access to received signals from the current time slot. 
An alternative motivation for this type of instantaneous relay may be found in \cite{LeeWang2013}.
Mathematically, we may describe non causal and causal relaying functions, for each $k=1,2,\cdots n$,  as
\begin{align*}
&\mbox{Non-causal / instantaneous relaying: }\\
& \;\;\;\;\;\;\;\; \mathbf{X_R}[k]=g_k(\mathbf{Y_R}[1], \mathbf{Y_R}[2],...,\mathbf{Y_R}[k]) \\
&\mbox{Causal relaying: }\\
& \;\;\;\;\;\;\;\; {\bf X_R}[k] = g_k({\bf Y_R}[1], {\bf Y_R}[2], \cdots {\bf Y_R}[k-1]),
\end{align*}
where $\mathbf{X_R}[k]$ is a $M \times 1$ ($M$ antennas) vector signal transmitted by the relay at time slot $k$; $g_k()$ is a deterministic function; and $\mathbf{Y_R}[l], l\in \{1,2,...,k\}$ is the  $M \times 1$ vector of signals received at the relay  at time slot $l$. The relay is subject to per symbol transmit power constraints over all antennas $E[||\mathbf{X_R}[k]||_2^2]\leq P_R$, $\forall k\in \{1,2,\cdots n\}$, and global channel state information knowledge is assumed at all nodes. 
%
%
%
%
%
At each time slot $k$, the system input/output relationships are:
\begin{align}
&Y_p[k]=\sum_{m=1}^Kh_{2m,p}[k]X_{2m}[k]+\mathbf{h_{Rp}^*}[k]\mathbf{X_R}[k]+Z_p[k], \label{snr1}\\
& \hspace{5cm} p=1,3,...,2K-1 \nonumber \\
&Y_q[k]=\sum_{m=1}^Kh_{2m-1,q}[k]X_{2m-1}[k]+\mathbf{h_{Rq}^*}[k]\mathbf{X_R}[k]+Z_q[k], \label{smr2}\\
& \hspace{5cm} q=2,4,...,2K \nonumber \\ 
& \mathbf{Y_R}[k]=\sum_{m=1}^{2K}\mathbf{h_{m,R}}[k]{X_m}[k]+\mathbf{Z_R}[k]\label{smr3}
\end{align}

where we use the same notation as in \eqref{sm1} and \eqref{sm2}. In addition, $\mathbf{h_{ij}}[k], i,j \in \{1,2,...,2K,R\}$ is the $M\times 1$-dimensional channel coefficient vector from node $i$ to node $j$ at time slot $k$ ($i$ or $j$ must be the relay node $R$), and $\mathbf{Z_R}[k]\sim \mathcal{CN} (\mathbf{0},\mathbf{I})$ is the complex Gaussian noise vector at the relay. The terms in bold represent vectors (due to the MIMO relay). We use $^*$ to denote conjugate transpose and $^T$ to denote transpose.

 \subsection{Types of signals}
 
Let $s_{ij}$ denote the independent information symbols (signals) from  transmitter $i$ to receiver $j$; these are real numbers which will be combined into the signals $X_i[k]$ transmitted by node $i$ at channel use $k$. 
The received signal at any given node may be broken down into four types of signals: 
\begin{itemize}
\item the  self-interference signal (SI, sent by itself, known to itself);
\item the interference signal (sent by the undesired user(s) from the opposite side);
\item  the desired signal (sent by the desired user);
\item the undesired signal (sent by the undesired user(s) from the same side), respectively. 
\end{itemize}
For example, at receiver 1, $s_{12}$ is a self-interference signal (SI); $s_{43}, s_{65}, \cdots, s_{2K, 2K-1}$ are the interference signals; $s_{21}$ is the desired signal, and $s_{34}, s_{56}, \cdots, s_{2K-1,2K}$ are the undesired signals. 
Note that we differentiate between interference and undesired signals (both of which in fact do not carry any messages desired by a particular node) as they will be treated in different ways in the following: interference signals may originate from other users (via direct links) or the relay and neutralized (combined with the direct links over which they are received to cancel the interference) by choice of the relay beam forming vector, while undesired signals would be received from the relay node only, but will be nulled by proper choice of beam forming vectors at the relay.


Note we have already removed self-interference signals from the input/output equations \eqref{sm1}-\eqref{smr2}, but SI terms may still be transmitted by the relay (or other users due to adaptation) and hence received. 

\section{DoF of $K$-pair-user two-way IC}
\label{KDoF}
In this section we show that the degrees of freedom of the $K$-pair-user two-way IC is $K$, i.e. $K/2$ in each direction (the DoF of a one-way $K$-user IC is $K/2$ \cite{Jafar:2008:alignment}), for both time-varying and  (almost all) constant channel coefficients. This result indicates that while the full-duplex operation essentially doubles the DoF,  interaction between users cannot further increase the DoF beyond what full-duplex allows. 
This may be intuitively explained as follows: 1) the DoF measures the number of clean information streams that may be transmitted at high SNR when the desired signals and interference signals are received roughly ``at the same level'' (SNR and INR scale to infinity at the same rate).  In this case, rates cannot be improved by having users send messages of other users to re-route the message (i.e. message from user 1 to 2 could go via 1 to 4, 4 to 3 then 3 to 2 instead) as all links are equally strong. One would thus need to tradeoff one's own rate to relay another user's rate given the symmetry in the channels.  2) Adaptation allows for the correlation of messages at transmitters. In Gaussian channels, such correlation may be translated into coherent power gains inside the logarithm. The DoF metric is insensitive to coherent power gains as it measures pre-logarithm gains, not constant power factor gains inside the logarithm, and hence even correlation between inputs which adaptation/interaction would permit does not improve the DoF.

The main result of this section is stated in the following theorem:

\smallskip 
\begin{theorem}
\label{kdof}
The full-duplex $K$-pair-user two-way interference channel has $K$ degrees of freedom.
\end{theorem}

\begin{proof}
We use the achievability scheme used in demonstrating the DoF of the one-way $K$-user IC (\cite{Jafar:2008:alignment} for time-varying channel and \cite{Motahari2009} for (almost all) constant channels) in each $\rightarrow$ and $\leftarrow$ direction simultaneously with non-interactive nodes (i.e. each direction uses this scheme and ignores the past received signals, a non-interactive scheme).  By making the appropriate correspondences, we see that $K/2$ DoF are achievable in each direction, leading to a sum DoF of $K$. This assumes that self-interference is able to be perfectly cancelled. 

Now we prove the converse, which is valid for both time-varying and constant channel gains. This outer bound carefully merges techniques used in the one-sided or Z interference channel \cite{costa_interference, sason, Sato:degraded, etkin_tse_wang} (asymmetric genies to the two receivers in one direction), in the symmetric $K$-user interference channel with feedback  \cite{Mohajer2012} (providing differences between noises as genies), and in determining the DoF of the $K$-user interference channel \cite{Jafar:2008:alignment} (re-scaling of channel coefficients). This combination is new and relies on novel constructions, and is more involved than the individual pieces given the larger number of messages and noises  involved and the fact that we allow for adaptation. This bound can also be extended to the symmetric Gaussian channel model and used to show a constant gap to capacity result as in \cite{Cheng:GlobalSIP} (i.e. adaptation or interaction, in some regimes of the symmetric Gaussian channel ,can only improve capacity to within an additive constant gap).

First rewrite channel outputs in \eqref{sm2} for $q=6,8,...,2K$ as
\begin{align}
Y_q^\prime [k]=h_{14}[k]X_1[k]+\sum_{m=2}^K\frac{h_{14}[k]}{h_{1,q}[k]}h_{2m-1,q}[k]X_{2m-1}[k]+Z_q^\prime [k]\label{rewrite}
\end{align}
where $Z_q^\prime [k]\sim \mathcal{CN} (0,\frac{h_{14}^2[k]}{h_{1,q}^2[k]})$. Let $M_A$ denote all the messages except $M_{12}, M_{34}$, and let $Z_{3,...,2K-1}[k]$ denote noises $Z_3[k],Z_5[k],...,Z_{2K-1}[k]$. Define $\bar{Z}_q[k]=Z_q^\prime [k]-Z_4[k], q=6,8,...,2K$, which is ${\cal CN}(0, 1+\frac{h_{14}^2[k]}{h_{1,q}^2[k]})$. 
Let $\bar{Z}_{6,...,2K}[k]$ denote $\bar{Z}_6[k], \bar{Z}_8[k],...,\bar{Z}_{2K}[k]$. 
We start by bounding the sum of a pair of rates: 
\begin{align}
&n(R_{12}+R_{34}-\epsilon) \nonumber\\
&\overset{(a)}{\leq} I(M_{34};Y_4^n|M_A, Z_{3,...,2K-1}^n)\nonumber\\
& \;\;\;\;\;\;\;\; +I(M_{12};Y_2^n,Y_4^n|M_{34},M_A, Z_{3,...,2K-1}^n)\nonumber\\
&\overset{(b)}{\leq} I(M_{34};Y_4^n,\bar{Z}_{6,...,2K}^n |M_A,Z_{3,...,2K-1}^n) \nonumber\\
& \;\;\;\;\;\;\;\; +I(M_{12};Y_2^n,Y_4^n,\bar{Z}_{6,...,2K}^n|M_{34},M_A,Z_{3,...,2K-1}^n)\nonumber
\end{align}
\begin{align}
&=H(Y_4^n,\bar{Z}_{6,...,2K}^n|M_A,Z_{3,...,2K-1}^n)\nonumber \\
& \;\;\;\;\;\; -H(Y_4^n,\bar{Z}_{6,...,2K}^n|M_{34},M_A,Z_{3,...,2K-1}^n)\nonumber\\
&\;\;\;\;\;\;\;\;+H(Y_2^n,Y_4^n,\bar{Z}_{6,...,2K}^n|M_{34},M_A,Z_{3,...,2K-1}^n)\nonumber\\
& \;\;\;\;\;\;\;\;\;\; -H(Y_2^n,Y_4^n,\bar{Z}_{6,...,2K}^n|M_{12},M_{34},M_A,Z_{3,...,2K-1}^n)\nonumber \\
&=H(Y_4^n,\bar{Z}_{6,...,2K}^n|M_A,Z_{3,...,2K-1}^n)\nonumber\\
& \;\;\;\;\;\; +H(Y_2^n|Y_4^n,\bar{Z}_{6,...,2K}^n,M_{34},M_A,Z_{3,...,2K-1}^n)\nonumber\\
& \;\;\;\;\;\;\;\; -H(Y_2^n,Y_4^n,\bar{Z}_{6,...,2K}^n|M_{12},M_{34},M_A,Z_{3,...,2K-1}^n) \label{eq:bound}
\end{align}
where (a) follows as all messages and noises are independent of each other; (b) by adding the side information $\bar{Z}_{6,...,2K}^n $.
Now we bound the three terms above respectively. We start with the first term as in \eqref{eq:long1} -- \eqref{eq:longend}
\begin{align}
& H(Y_4^n,\bar{Z}_{6,...,2K}^n|M_A,Z_{3,...,2K-1}^n)\label{eq:long1}\\
&\leq H(Y_4^n)+H(\bar{Z}_{6,...,2K}^n)\\
&\overset{(a)}{\leq }  
n(\log(P) + \smallO{\log(P)})+H(\bar{Z}_{6,...,2K}[k])\\
&\leq n(\log(P) + \smallO{\log(P)})\\
& +\sum_{k=1}^n\log (2\pi e)^{K-2}\left(1+\frac{h_{14}^2[k]}{h_{16}^2[k]}\right)\cdots \left(1+\frac{h_{14}^2[k]}{h_{1,2K}^2[k]}\right)\nonumber \\
&=  n(\log(P) + \smallO{\log(P)}) \label{eq:longend}
\end{align}
where in (a) we have used the fact that Gaussians maximize entropy subject to power constraints (which we recall are $P$ at each user and time slot). Due to adaptation, the inputs $X_{2m-1}, m\in \{1,2,...,K\}$ may be correlated, but even if all users are fully correlated and all the transmitters meet the power constraint $P$, $H(Y_4^n) \leq n(\log P +\smallO{\log P})$ as correlation only induces a power gain inside the logarithm for a single antenna receiver.\footnote{We note that a bound of $n\log(P) + \smallO{\log(P)}$ may also be shown to hold for average rather than per symbol power constraints of $P$ at each transmitter using Jensen's inequality and using that $2\sqrt{P_i P_j} \leq P_i+P_j$. } 
Here $f(x) = \smallO{\phi(x)}$  denotes the Landau little-O notation, i.e. that  $\lim_{x\rightarrow \infty} \frac{f(x)}{\phi(x)} =0$.

The second term can be bounded as follows in \eqref{eq:long3}-\eqref{eq:frac2}.
\begin{figure*}
\begin{align}
&H(Y_2^n|Y_4^n,\bar{Z}_{6,...,2K}^n,M_{34},M_A,Z_{3,...,2K-1}^n)\label{eq:long3}\\
&\leq H(Y_2^n,Y_3^n,Y_5^n,...,Y_{2K-1}^n|Y_4^n,\bar{Z}_{6,...,2K}^n,M_{34},M_A,Z_{3,...,2K-1}^n)\nonumber\\
& =\sum_{k=1}^n[H(Y_2[k],Y_3[k],Y_5[k],...,Y_{2K-1}[k]|Y_2^{k-1},Y_3^{k-1},Y_5^{k-1},...,Y_{2K-1}^{k-1},Y_4^n,\bar{Z}_{6,...,2K}^n,M_{34},M_A,\nonumber\\
& X_2^{k},X_3^{k},X_5^{k},...,X_{2K-1}^{k},X_4^n,Z_{3,...,2K-1}^n)]\nonumber\\
&\overset{(c)}{=}\sum_{k=1}^n[H(Y_2[k],Y_3[k],Y_5[k],...,Y_{2K-1}[k]|Y_2^{k-1},Y_3^{k-1},Y_5^{k-1},...,Y_{2K-1}^{k-1},Y_4^n,\bar{Z}_{6,...,2K}^n,M_{34},M_A,\nonumber\\
& X_2^{k},X_3^{k},X_5^{k},...,X_{2K-1}^{k},X_4^n,Z_{3,...,2K-1}^n, X_6[k],X_8[k],...,X_{2K}[k])]\nonumber\\
&\leq \sum_{k=1}^n [H(h_{12}[k]X_1[k]+Z_2[k]|h_{14}[k]X_1[k]+Z_4[k])]\nonumber\\
&\overset{(d)}{\leq}\sum_{k=1}^n\log 2\pi e \left( 1+\frac{h_{12}^2[k]P}{1+h_{14}^2[k]P}\right) \label{eq:frac1} \\
&= n(\smallO{\log(P)}) \label{eq:frac2}
\end{align}
\end{figure*}
where in step (c) we construct $X_6[k]$ in the conditioning because (1), $h_{14}^{k-1} X_1^{k-1}+Z_4^{k-1}$ (this notation is meant to compactly represent the $k-1$ dimensional vector of the $k-1$ equations 
$h_{14}[l]X_1[l] + Z_4[l]$, for $l=1,2,\cdots k-1$) 
can be decoded from $Y_4^n$ since $X_3^{k},X_5^{k},...,X_{2K-1}^{k}$ are known; (2), $\bar{Z}_6^{k-1}$ is known so that $h_{14}^{k-1}X_1^{k-1}+Z_6^{\prime [k-1]}$ can be constructed; and (3), together 
with the knowledge of $X_3^{k-1},X_5^{k-1},...,X_{2K-1}^{k-1}$ and (1), (2),  $Y_6^{\prime [k-1]}$ is constructed; finally (4), perfect CSI at receivers, i.e. knowing $Y_6^{\prime [k-1]}$ is equivalent to knowing 
$Y_6^{k-1}$, and combining this with the knowledge of $M_{65}$, according to the interactive encoding function we can construct $X_6[k]$. Similarly $X_8[k],...,X_{2K}[k]$ can be constructed. (d) follows since Gaussians maximize conditional entropies, as in for example \cite[Equation (30)]{Jafar:2009:relays}.

Finally, the negative third term can be lower bounded as follows: 
\begin{align*}
&H(Y_2^n,Y_4^n,\bar{Z}_{6,...,2K}^n|M_{12},M_{34},M_A,Z_{3,...,2K-1}^n)\\
&\geq H(Y_2^n,Y_4^n,\bar{Z}_{6,...,2K}^n|M_{12},M_{34},M_A,Z_{3,...,2K-1}^n, \nonumber \\
& \hspace{4cm}X_1^n,X_3^n,X_5^n,..., X_{2K-1}^n)\\
&= \; H(Z_2^n,Z_4^n,\bar{Z}_{6,...,2K}^n)\\
&=H(Z_2^n,Z_4^n,Z_6^{\prime n},..., Z_{2K}^{\prime n})\\
&=n\log 2\pi e+n\log 2\pi e+\sum_{k=1}^n\log (2\pi e)^{K-2} \frac{h_{14}^2[k]}{h_{16}^2[k]}\cdots\frac{h_{14}^2[k]}{h_{1,2K}^2[k]} \\
&= n(\bigO{1}) \label{eq:noise},
\end{align*}
where $f(x) = \bigO{\phi(x)}$ denotes that $|f(x)|< A\phi(x)$ for some constant $A$ and all values of $x$. 
Now combining everything, and taking the limit, 
\begin{align}
& d_{12}+d_{34}\leq \limsup_{P\rightarrow\infty} \frac{R_{12}+R_{34}}{\log (P)} = 1+0+0-0=1.
\end{align}
From the above we see that the DoF per pair of users transmitting in the same direction is $1$. Summing over all rate pairs leads to the theorem. 
\end{proof}

\section{DoF of $K$-pair-user two-way IC with an instantaneous MIMO relay}
\label{2KDoF}

In this section, we investigate the DoF of the $K$-pair-user two-way IC with an instantaneous MIMO relay with $M=2K$ antennas in the system model described in Section \ref{modelrelay}. 
{We then make a number of comments on how to reduce the number of antennas at the relay, at the expense of for example diminished achievable degrees of freedom, or requiring partial cognition of the messages at the relay.}  
\subsection{DoF of $K$-pair-user Two-way IC with an instantaneous 2K-antenna Relay}

We show our second main result: that the maximum $2K$ DoF of the $K$-pair-user two-way IC with an instantaneous $2K$-antenna relay is achievable. 

\begin{theorem}
\label{2kdof}
The full-duplex $K$-pair-user two-way interference channel with an instantaneous $2K$-antenna relay has $2K$ degrees of freedom.
\end{theorem}
\begin{proof}
\subsubsection{Converse}
The converse is trivial since for a $2K$-user, $2K$ message unicast network where all sources and destinations have a single antenna, the maximum degrees of freedom cannot exceed $2K$ by cut-set arguments, even with adaptation/interaction at all nodes. 

\subsubsection{Achievability}

We  propose a simple ``one-shot'' scheme that achieves $2K$ DoF for the $K$-pair-user two-way IC with the help of an instantaneous $2K$-antenna relay. We consider the Gaussian channel model at high SNR, and hence noise terms are ignored from now on.

The $2K$ users each transmit a symbol $s_{ij}$ (from user $i$ to intended user $j$) and the relay receives:
\begin{align*}
\mathbf{Y_R}=\sum_{i=1}^{2K}\mathbf{h_{i,R}}s_{ij}, \mbox{ for the appropriate } j \mbox{ values, see Fig. \ref{fig:kuser}.}
\end{align*}
The $2K$-antenna relay (with global CSI) decodes all $2K$ symbols using a zero-forcing decoder \cite{Goldsmith:2005}, and due to the instantaneous property, transmits the following signal in the same time slot:\begin{align*}
\mathbf{X_R}=\sum_{i=1}^{2K}\mathbf{u_{ij}}s_{ij}
\end{align*} 
where $\mathbf{u_{ij}}$ denote the $2K\times 1$ beamforming vectors carrying signals from user $i$ to intended user $j$. Now at receiver 1 (for example),
\begin{align}
Y_1=\sum_{m=1}^Kh_{2m,1}s_{2m,2m-1}+\mathbf{h_{R1}^*}\mathbf{X_R}.\label{rx1}
\end{align}
To prevent undesired signals from reaching receiver 1, the relay picks beamforming vectors such that 
\begin{align}
\mathbf{u_{ij}}\in null (\mathbf{h_{R1}^*}), \ \ \ i=3,5,...,2K-1, \; j\mbox{ as appropriate,}\label{nullify}
\end{align}
where $null(\mathbf{A})$ denotes the null space of $\mathbf{A}$. Since there are $2K$ antennas at the relay, $null (\mathbf{h_{R1}^*})$ has dimension $2K-1$.

At receiver 1, the interference signals received from the relay are used to neutralize the interference signals received from the transmitters. To do this, we design the beamforming vectors to satisfy:\begin{align}
h_{2m,1}+\mathbf{h_{R1}^*}\mathbf{u_{2m,2m-1}}=0, \ \ \ m=2,3,...,K. \label{intf-ntrlz}
\end{align} 

The $2K\times 1$ beamforming vectors satisfying the needed constraints always exist, by a dimensionality argument, along with the random channel coefficients. To see this, take $\mathbf{u_{34}}$ as an example. We wish to construct $\mathbf{u_{34}}$ such that the following conditions are satisfied:
\begin{align}
&\mathbf{u_{34}}\in null(\mathbf{h_{Rp}^*}), \ \ \ p=1,5,7,...,2K-1\label{null1}\\
&h_{3q}+\mathbf{h_{Rq}^*}\mathbf{u_{34}}=0, \ \ \ q=2,6,8,...,2K.\label{ntr1}
\end{align}  
From $\mathbf{u_{34}}\in null (\mathbf{h_{R1}^*})$ (one condition in \eqref{null1} for $p=1$), we see that there are $2K-1$ free parameters, which are reduced to 2 in order to satisfy the other $K-2$ conditions in \eqref{null1} for $p=5,7,\cdots 2K-1$, and the $K-1$ conditions in \eqref{ntr1}. That is, $(2K-1)-(K-2)-(K-1)=2$. Thus, let $a,b$ be two scalars, let ${\bf A,B}$ be $1\times 2K$ vectors such that the matrix below is invertible, then the following choice of beam forming vector (for example) will satisfy all conditions:
\begin{align}
&\mathbf{u_{34}}=\begin{bmatrix} \mathbf{h_{R1}^*} \\ \mathbf{h_{R2}^*} \\ \mathbf{h_{R5}^*} \\ \mathbf{h_{R6}^*} \\ \vdots \\ \mathbf{h_{R,2K}^*} \\ \mathbf{A} \\ \mathbf{B}\end{bmatrix}^{-1}\begin{bmatrix} 0 \\ -h_{32} \\ 0 \\ -h_{36} \\ \vdots \\ -h_{3,2K} \\ a \\ b \end{bmatrix}. \label{sbfv}
\end{align} 
Note that all the beam forming vectors must also be chosen to satisfy the relay power constraint $P_R$, but that we have sufficient degrees of freedom (choices of a,b) to ensure this, and that this will not affect the DoF in either case, as we will let $P_R\rightarrow \infty$, essentially removing the power constraint.

Still at receiver 1, once the interference signals have been neutralized and the undesired signals have been nulled  (by the above choice of beam forming vectors), and   the self-interference (SI) signal $s_{12}$ has been subtracted off, the received signal in \eqref{rx1} becomes
\begin{align}
Y_1-SI=h_{21}s_{21}+\mathbf{h_{R1}^*}\mathbf{u_{21}}s_{21},
\end{align}
from which the desired signal $s_{21}$ can be easily decoded as long as $h_{21} \neq -{\bf h_{R1}^*}{\bf u_{21}}$, which we may guarantee by proper scaling of ${\bf u_{21}}$. Similar decoding procedures are performed at all other receivers. Note that we have again assumed that self-interference in the full-duplex system is able to be perfectly removed. This provides an ideal upper bound to what is currently realizable, and including the effect of self-interference on rates is beyond the scope of this work, but an interesting topic for future investigation.
\end{proof}

\begin{remark} To achieve $2K$ DoF we have assumed full duplex operation. If instead all nodes operate in half-duplex mode, intuitively the DoF will be halved, i.e. $K$. Indeed, it is trivial to achieve $K$ DoF in  a half-duplex setup: In the first time slot, all $2K$ users transmit a message, and the $2K$-antenna relay listens and  decodes all $2K$ messages using a zero-forcing decoder. At time slot 2, the relay broadcasts a signal and all users listen. {By careful choice of  beamforming vectors as in \eqref{sbfv}, for example,  each receiver receives only their desired message in this time slot.} Therefore $2K$ desired messages are obtained in 2 time slots, i.e. $K$ DoF is achievable. Note however that in the half-duplex setting, the relay is causal rather than non-causal or instantaneous. 
\end{remark}

\begin{remark} We have shown in the previous section that the DoF of the $K$-pair-user two-way interference channel is $K$; Theorem \ref{2kdof} implies that  the addition of an instantaneous $2K$-antenna relay can increase the DoF of the $K$-pair-user two-way IC to $2K$ -- it essentially cancels out all interference in both directions simultaneously. This may have interesting design implications for full duplex two-way interference networks -- i.e. the ability of nodes to operate in full duplex would double  the DoF of the two-way $K$-pair user IC from $K/2$ (each direction time-shares) to $K$; the addition of a full-duplex, instantaneous MIMO relay with $2K$ antennas (for example, a pico-cell) would again double this to $2K$ DoF. 
\end{remark}

\subsection{Comments on reducing the number of antennas at the instantaneous relay}
\label{subsec:antennas}

We now investigate how many DoF can be achieved by using a reduced number of antennas at the instantaneous relay. For simplicity, we first consider the (2-pair-user) two-way IC with an instantaneous 3-antenna relay, for which we propose another one-shot strategy to achieve 3 DoF.  Whether this is the optimal achievable DoF is still open, i.e. unlike in all other sections so far, we have not obtained a  converse. 

\begin{theorem}
For the full-duplex two-way interference channel with an instantaneous 3-antenna relay, 3 degrees of freedom are achievable.
\end{theorem}

\begin{proof}

We now demonstrate how to transmit 3 symbols (one for each of three users) in one time slot using a 3 antenna relay. Then, this clearly achieves 3 DoF, which is larger than the 2 achievable without a relay, but smaller than the maximal value of 4 (whether anything larger than 3 is achievable is left open). 

Let transmitter 1,2 and 3 transmit symbols $s_{12}, s_{21}$ and $s_{34}$. Transmitter 4 stays silent. The relay, with three antennas is then able to use a zero-forcing receiver to obtain the three transmitted symbols, and then proceeds to transmit 
\begin{align*}
\mathbf{X_R}&=\mathbf{u_{12}}s_{12}+\mathbf{u_{21}}s_{21} + \mathbf{u_{34}}s_{34}.
\end{align*}  
The receivers 1,2 and 4 (since transmitter 4 is not sending anything, receiver 3 has no desired message and we ignore it) then receive the signals:
\begin{align*}
Y_1 & = h_{21}s_{21}+{\bf h_{R1}^*}(\mathbf{u_{12}}s_{12}+\mathbf{u_{21}}s_{21} + \mathbf{u_{34}}s_{34})\\
Y_2 & = h_{12}s_{12}+  h_{32}s_{34}+{\bf h_{R2}^*}(\mathbf{u_{12}}s_{12}+\mathbf{u_{21}}s_{21} + \mathbf{u_{34}}s_{34})\\
Y_4 & = h_{34}s_{34}+h_{14}s_{12}+{\bf h_{R4}^*}(\mathbf{u_{12}}s_{12}+\mathbf{u_{21}}s_{21} + \mathbf{u_{34}}s_{34}).
\end{align*}
At receiver 1, to decode the desired $s_{21}$, we first subtract off the self-interference term ${\bf h_{R1}^*}{\bf u_{12}}s_{12}$ and then design ${\bf u_{34}}$ such that the undesired term in $s_{34}$ disappears, i.e. make
\begin{align}
 {\bf u_{34}} \in null({\bf h_{R1}^*)}. \label{u341}
 \end{align}
 Then receiver 1 is able to decode $s_{21}$ as long as $h_{21}+{\bf h_{R1}^*}{\bf u_{21}} \neq 0$, which may be guaranteed by proper scaling of ${\bf u_{21}}$.
 
 At receiver 2, to decode the desired $s_{12}$, we first subtract off the self-interference term ${\bf h_{R2}^*}{\bf u_{21}}s_{21}$, then neutralize the interference term from $s_{34}$ by selecting ${\bf u_{34}}$ such that
\begin{align}
 h_{32} + {\bf h_{R2}^*}{\bf u_{34}} = 0. \label{u342}
 \end{align} 
 Then receiver 2 is able to decode $s_{12}$ as long as $h_{12}+{\bf h_{R2}^*}{\bf u_{12}} \neq 0$, which may be guaranteed by proper scaling of ${\bf u_{12}}$. 
 
Finally, at receiver 4, there is no self-interference term, only the desired term in $s_{34}$, an interference term in $s_{12}$ and an undesired signal term $s_{21}$.
We may nullify the undesired signal term by selecting
\begin{align}
 {\bf u_{21}} \in null({\bf h_{R4}^*)}. \label{u21}
 \end{align}
Then, select ${\bf u_{12}}$ to neutralize the interference by selecting
\begin{align}
 h_{14} + {\bf h_{R4}^*}{\bf u_{12}} = 0. \label{u12}
 \end{align} 
 Then receiver 2 is able to decode $s_{34}$ as long as $h_{34}+{\bf h_{R4}^*}{\bf u_{34}} \neq 0$, which may be guaranteed by proper scaling of ${\bf u_{34}}$. 
 
Each ${\bf u_{12}, u_{21}, u_{34}}$ is a $3\times 1$ vector. There is one linear constraint on ${\bf u_{12}}$, one linear constraint on ${\bf u_{21}}$, and two linear constraints on ${\bf u_{34}}$. Hence, we have enough degrees of freedom to select all beamforming vectors to satisfy the constraints, and hence achieve 3 DoF in one time slot.
 
We note that this scheme sends one symbol for three of the four users, i.e. the rate $R_{43}=0$ as no message is sent by transmitter 4. Though it does not matter from a DoF perspective (as this is defined as a sum of rates), one may symmetrize the rates by having different subsets of users transmit over 4 time slots. That is, in time slot 1, users 1,2,3 transmit. In time slot 2, users 1,2,4 transmit. In time slot 3, users 1,3,4 transmit, and in time slot 4 users 2,3,4 transmit. In this case, 12 symbols will be transmitted over 4 time slots, and each of the 4 users is able to transmit (or receive) 3 signals in 4 time slots, again leading to 3 DoF. 

\end{proof}

{The above result demonstrates that by reducing the number of antennas at the instantaneous relay from 4 to 3, we have also reduced the achievable DoF from 4 to 3. However, we want to point out that we do not currently have a converse, and more than 3 DoF may still be achievable (but clearly no more than the maximal 4). One may ask how else we might be able to reduce the number of antennas without impacting or decreasing the DoF.  One way is to trade cognition for antennas, as we remark on next. }

\begin{remark}
If we consider a {\it cognitive} relay (cognitive in the sense of having a-priori knowledge of messages, as first introduced in \cite{devroye_IEEE}), which would have access to {\it all} 4 users' signals prior to transmission, the number of antennas at the relay can be reduced to 2, while still being able to achieve the maximum 4 degrees of freedom for the full-duplex two-way IC. The achievability scheme is trivial: the cognitive relay broadcasts all 4 signals (desired for each user) and all users listen. By careful choice of the four $2\times 1$ beamforming vectors to nullify undesired and interference signals, and subtracting the self-interference signal, each receiver is able to obtain the desired signal. Therefore the maximal 4 DoF are achieved.    
\end{remark}

\begin{remark}
We can do even better: if the cognitive relay only knows 2 users' signals, then we are still able to achieve the maximum 4 DoF with 2 antennas at the relay by a simple one-shot scheme. For example, assume user 1 and 2's signals are known at the relay prior to transmission (knowing any 2 of the 4 messages suffices).
Now, each transmitter sends a message $s_{ij}$ and the relay receives 4 messages. Then the 2-antenna relay first subtracts transmitter 1 and 2's messages and uses a zero-forcing decoder to decode the other two messages, and transmits 
\begin{align*}
\mathbf{X_R}&=\mathbf{u_{12}}s_{12}+\mathbf{u_{21}}s_{21}+\mathbf{u_{34}}s_{34}+\mathbf{u_{43}}s_{43}.
\end{align*}  
At receiver 1 (for example):
\begin{align*}
Y_1&=h_{21}s_{21}+h_{41}s_{43}\\
&+\mathbf{h_{R1}^*}\mathbf{u_{21}}s_{21}+\mathbf{h_{R1}^*}\mathbf{u_{43}}s_{43}+\mathbf{h_{R1}^*}\mathbf{u_{12}}s_{12}+\mathbf{h_{R1}^*}\mathbf{u_{34}}s_{34}.
\end{align*}
To decode the desired message $s_{21}$, we subtract off the self-interference signal $s_{12}$; nullify the undesired signal $s_{34}$ by designing the beamforming vector such that $\mathbf{h_{R1}^*}\mathbf{u_{34}}=0$; and neutralize the interference signal $s_{43}$ by setting $h_{41}+\mathbf{h_{R1}^*}\mathbf{u_{43}}=0$. A similar decoding procedure follows for the other receivers, where we note the $2\times 1$ beamforming vectors can be always constructed by the 2-antenna relay. Therefore, each user is able to get 1 desired signal in 1 time slot and the maximal 4 DoF are achieved.
\end{remark}

\section{DoF of K-pair-user two-way interference channel with a causal MIMO relay}
\label{sec:causal}

It is known that for one-way channels where nodes are either sources of destinations of messages but not both as in a two-way setting, the usage of feedback, causal relays (possibly with multiple antennas), and cooperation does not increase the DoF of the network \cite{Jafar:2009:relays}. In the previous section, we showed that a non-causal / instantaneous multi-antenna relay may increase the DoF of a two-way $K$-pair user interference channel to its maximal value of $2K$ (provided we have sufficient number of antennas). Here we show that, in sharp contrast, if the relay is actually causal, it does {\it not} increase the DoF of the $K$-pair-user two-way IC beyond that of a network without the relay present, which would have $K$ DoF ($K/2$ in each direction). This aligns with (and the proof uses similar techniques) the one-way results in \cite{Jafar:2009:relays} in the sense that causal relays again do not help. However, we note that full-duplex operation {\it does} increase the DoF for the two-way networks in this paper, but does not for their one-way counterparts \cite{Jafar:2009:relays}.

We thus consider a $K$-pair-user two-way IC with one causal MIMO relay which has $M$ antennas. The system model is the same as that in Section \ref{modelrelay}, where we recall that the relay is now causal, and hence 
\[ {\bf X_R}[k] = g_k({\bf Y_R}[1], {\bf Y_R}[2], \cdots {\bf Y_R}[k-1]), \]
where ${\bf X_R}[k]$ is an $M\times 1$ vector signal transmitted by the relay at time $k$, $g_k()$ are deterministic functions for each $k=1,2,\cdots n$, and ${\bf Y_R}[k]$ is the $M\times 1$ vector of signals received by the relay at time slot $k$.  Let $P=P_R$ for simplicity (we simply need $P$ and $P_R$ to scale to infinity at the same rate). Our third main result is the following.

\smallskip

\begin{theorem}
The DoF of the $K$-pair-user full-duplex two-way interference channel with a causal MIMO relay is $K$.
\end{theorem}

\begin{proof}
Achievability follows from the fact that the DoF of the $K$-pair-user two-way interference channel without a relay is $K$, as shown in Section \ref{KDoF}. 

\begin{figure}
\begin{center}
\includegraphics[width=3.7in]{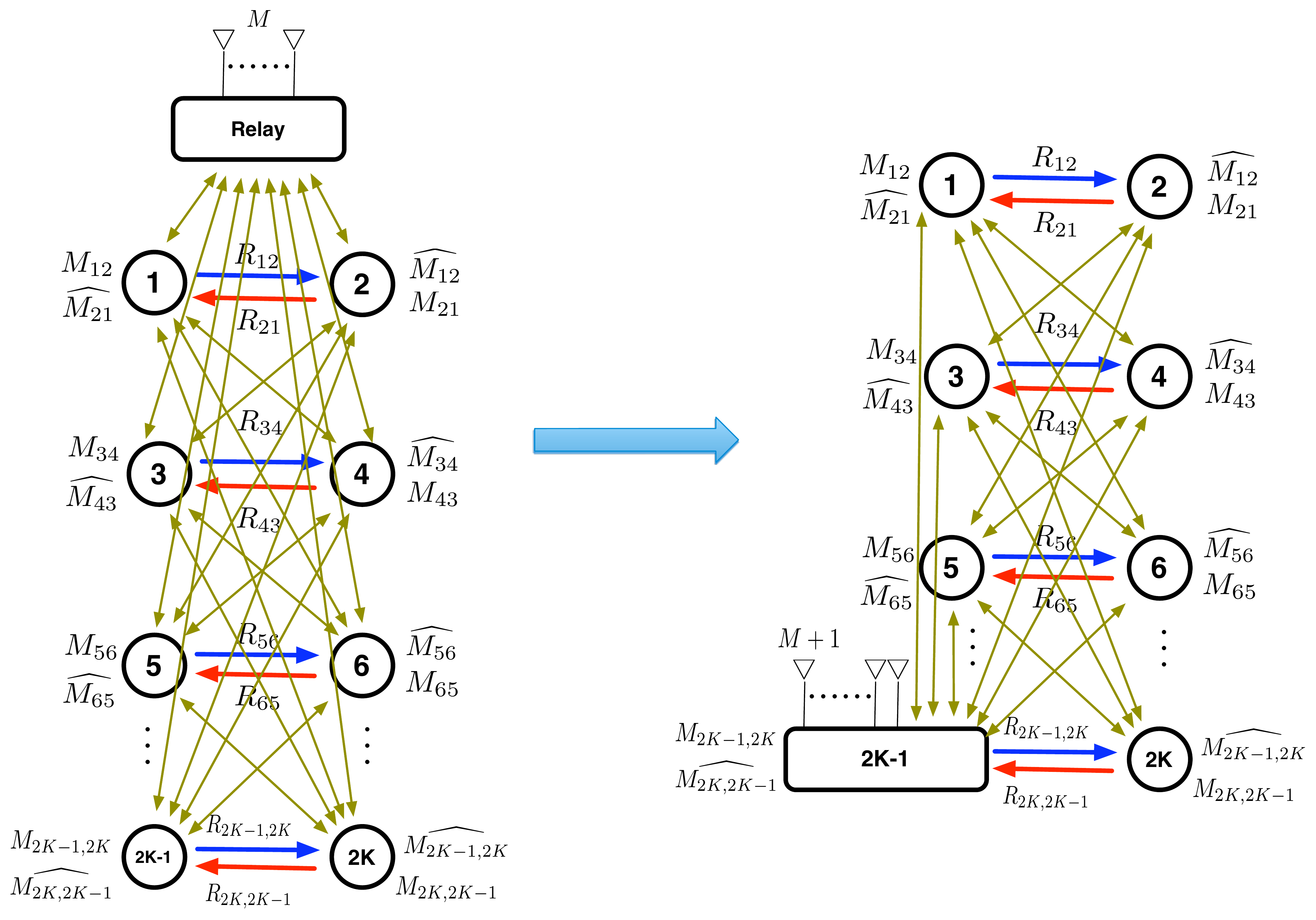}
\caption{Transformation of the K-pair-user full-duplex two-way interference channel with a causal MIMO relay.}
\label{fig:transform}
\end{center}
\end{figure}
Now we prove the converse. Inspired by \cite{Jafar:2009:relays}, we first transform our $2K+1$ node network to a $2K$-node network as shown in Fig. \ref{fig:transform}.
Since cooperation between nodes cannot reduce the DoF, 
  we let the causal MIMO relay fully cooperate with one of the users, take user $2K-1$ WLOG. In other words, we co-locate user $2K-1$ and the relay or put infinite capacity links between these nodes. Then the capacity region of the original network is outer bounded by that of the following $2K$-node network which each have one message and desire 1 message as before: All users except user $2K-1$ each have a single antenna, while user $2K-1$ has $M+1$ antennas (one from the original node $2K-1$, and $M$ from the relay). Since the original relay is connected to all $2K$ users, user $2K-1$ in the transformed network is connected to all other users, in contrast to the original network where there is no direct link between users 1,3,...,$2K-3$ and $2K-1$. Then, letting  the tilde $\tilde{A}$ notation denote the inputs, outputs and channel gains of the new network, we have the correspondences (or equivalences $\equiv$ for inputs, since they may actually be different due to interaction based on different received signals)
  \begin{align*}
 &\tilde{X_i}  \equiv X_i,  \;\; i=1,2,...,2K, \mbox{except} \ 2K-1, \;\; {\bf  \tilde{X_{2K-1}}}^T   \equiv [X_{2K-1}, {\bf X_R}^T], \\
 &\tilde{Z_i}  \equiv Z_i,  \;\; i=1,2,...,2K, \mbox{except} \ 2K-1, \;\; {\bf  \tilde{Z_{2K-1}}}^T   \equiv [Z_{2K-1}, {\bf Z_R}^T], \\
&  \tilde{h_{ij}}  = h_{ij}, \;\; \mbox{for appropriate} \ i,j \ \mbox{and} \ i,j\neq 2K-1 \\
&\tilde{{\bf h_{i,2K-1}}}^T  = [0, {\bf h_{iR}}^T], \;\; i=1,3,...,2K-3,  \\
&\tilde{{\bf h_{i,2K-1}}}^T  = [h_{i,2K-1}, {\bf h_{iR}}^T], \;\; i=2,4,...,2K,  \\
&\tilde{{\bf h_{2K-1,j}}}^T  = [0, {\bf h_{Rj}}^T],  \;\; j=1,3,...,2K-3, \\
&\tilde{{\bf h_{2K-1,j}}}^T  = [h_{2K-1,j}, {\bf h_{Rj}}^T],  \;\; j=2,4,...,2K,
  \end{align*}
and the following input/output relationships at each channel use:
\begin{align}
&\tilde{Y_p}[k]=\sum_{m=1}^K\tilde{h_{2m,p}}[k]\tilde{X_{2m}}[k]+\mathbf{\tilde{h_{2K-1,p}^*}}[k]\mathbf{\tilde{X_{2K-1}}}[k]+\tilde{Z_p}[k], \nonumber \\
& \hspace{4.5cm}  p=1,3,...,2K-3\\
&\tilde{Y_q}[k]=\sum_{m=1}^{K-1}\tilde{h_{2m-1,q}}[k]\tilde{X_{2m-1}}[k]+\mathbf{\tilde{h_{2K-1,q}^*}}[k]\mathbf{\tilde{X_{2K-1}}}[k]+\tilde{Z_q}[k], \nonumber \\ & \hspace{4.5cm} q=2,4,...,2K\\
&\mathbf{\tilde{Y_{2K-1}}}[k]=\sum_{m=1,m\neq 2K-1}^{2K}\mathbf{\tilde{h_{m,2K-1}}}[k]\tilde{X_m}[k]+\mathbf{\tilde{Z_{2K-1}}}[k].
\end{align}
We have the  interactive encoding functions at each node
\begin{align}
&\tilde{X_i}[k]=\tilde{f_i}(M_{ij},\tilde{Y_i}^{k-1}),  \ \ \ \ \ i=1,2,...,2K, \mbox{except} \ 2K-1\\
&\mathbf{\tilde{X_{2K-1}}}[k]=\tilde{f_{2K-1}}(M_{2K-1,2K},\mathbf{\tilde{Y_{2K-1}}^{k-1}})\label{causal}
\end{align}
where \eqref{causal} is where the causality of the relay is observed / incorporated.


Let $M_A$ denote all the messages except $M_{12},M_{34}$, and let $\tilde{Y}_{(2,...,2K)/4}$ denote $\tilde{Y_2},\tilde{Y_3},\tilde{Y_5},...,\tilde{Y_{2K}}$ i.e. all outputs except $\tilde{Y_1}$ and $\tilde{Y_4}$. Note $\tilde{Y}_{(2,...,2K)/4}$ includes the outputs vector $\mathbf{\tilde{Y}_{2K-1}}$ at user $2K-1$. Similarly, $\tilde{X}_{(2,...,2K)/4}$ and $\tilde{Z}_{(2,...,2K)/4}$ denote all inputs and noises except those at nodes 1 and 4.

We now bound the sum-rate in each direction, considering the sum of a pair of rates,  and starting with Fano's inequality, we will have
\begin{align*}
&n(R_{12}+R_{34}-\epsilon) \\
&\leq I(M_{34};\tilde{Y_4}^n|M_A)+I(M_{12};\tilde{Y_4}^n,\tilde{Y}^n_{(2,...,2K)/4}|M_{34},M_A)\\
&\leq H(\tilde{Y_4}^n|M_A)-H(\tilde{Y_4}^n|M_{34},M_A)\\
& \;\; +H(\tilde{Y_4}^n,\tilde{Y}^n_{(2,...,2K)/4}|M_{34},M_A)\\
& \;\;\;\; -H(\tilde{Y_4}^n,\tilde{Y}^n_{(2,...,2K)/4}|M_{34},M_A,M_{12})\\
& =H(\tilde{Y_4}^n|M_A)+H(\tilde{Y}^n_{(2,...,2K)/4}|\tilde{Y_4}^n,M_{34},M_A) \\
& \;\;\;\; -H(\tilde{Z_4}^n,\tilde{Z}^n_{(2,...,2K)/4})\\
&= H(\tilde{Y_4}^n|M_A)-H(\tilde{Z_4}^n)+H(\tilde{Y}^n_{(2,...,2K)/4}|M_{34},M_A,  \tilde{Y_4}^n)\\
& \;\;\;\; -H(\tilde{Z}^n_{(2,...,2K)/4})\\
&\leq \sum_{k=1}^n [H(\tilde{Y_4}[k])-H(\tilde{Z_4}[k])\\
&+H(\tilde{Y}_{(2,...,2K)/4}[k]|\tilde{Y}_{(2,...,2K)/4}^{k-1},M_{34},M_A,\tilde{Y_4}^n, \tilde{X_4}^n,\tilde{X}_{(2,...,2K)/4}^{k})\\
& \;\;\;\;\;\;\;\;\;\; -H(\tilde{Z}_{(2,...,2K)/4}[k])]\\
&\leq n(\log P +o(\log P))\\
&+\sum_{k=1}^n[H(\tilde{h_{12}}[k]\tilde{X_1}[k]+\tilde{Z_2}[k], \tilde{Z_3}[k],\cdots, \mathbf{\tilde{h_{1,2K-1}}}[k]\tilde{X_1}[k] \\
& \;\;+\mathbf{\tilde{Z_{2K-1}}}[k],\tilde{h_{1,2K}}[k]\tilde{X_1}[k]+\tilde{Z_{2K}}[k]|\tilde{h_{14}}\tilde{X_1}[k]+\tilde{Z_4}[k])\\
&\;\;\;\; -H(\tilde{Z}_{(2,...,2K)/4}[k])]\\
&\leq n(\log P +o(\log P))+no(\log P),
\end{align*}
where the last step follows as it may be shown that the Gaussian distribution maximizes conditional entropy, as done in \cite[Equation (30), (31)]{Jafar:2009:relays}, similar to \cite[Lemma 1]{host_madsen_int}, and similar to \eqref{eq:frac1}, \eqref{eq:frac2}. Note also that the conditional entropy term involves a single-input, multiple output term, and hence is again bounded by $no(\log P)$, due to the conditioning.

Similarly, in the opposite direction, let $M_B$ denote all the messages except $M_{21},M_{43}$: 
\begin{align*}
&n(R_{21}+R_{43}-\epsilon) \\
&\leq I(M_{21};\tilde{Y_1}^n|M_B)+I(M_{43};\tilde{Y_1}^n,\tilde{Y}^n_{(2,...,2K)/4}|M_{21},M_B)\\
&\leq H(\tilde{Y_1}^n|M_B)-H(\tilde{Y_1}^n|M_{21},M_B)\\
& \;\;+H(\tilde{Y_1}^n,\tilde{Y}^n_{(2,...,2K)/4}|M_{21},M_B)\\
& \;\;\;\; -H(\tilde{Y_1}^n,\tilde{Y}^n_{(2,...,2K)/4}|M_{21},M_B,M_{43})
\end{align*}
\begin{align*}
& =H(\tilde{Y_1}^n|M_B)+H(\tilde{Y}^n_{(2,...,2K)/4}|\tilde{Y_1}^n,M_{21},M_B)\\
& \;\;-H(\tilde{Z_1}^n,\tilde{Z}^n_{(2,...,2K)/4})\\
&= H(\tilde{Y_1}^n|M_B)-H(\tilde{Z_1}^n)+H(\tilde{Y}^n_{(2,...,2K)/4}|M_{21},M_B,  \tilde{Y_1}^n)\\
& \;\;\;\; -H(\tilde{Z}^n_{(2,...,2K)/4})\\
&\leq \sum_{k=1}^n [H(\tilde{Y_1}[k])-H(\tilde{Z_1}[k])\\
&+H(\tilde{Y}_{(2,...,2K)/4}[k]|\tilde{Y}_{(2,...,2K)/4}^{k-1},M_{21},M_B,\tilde{Y_1}^n, \tilde{X_1}^n,\tilde{X}_{(2,...,2K)/4}^{k})\\
& \;\;\;\;\;\;\;\;\;\; -H(\tilde{Z}_{(2,...,2K)/4}[k])]\\
&\leq n(\log P +o(\log P))\\
&+\sum_{k=1}^n[H(\tilde{Z_2}[k], \tilde{h_{43}}[k]\tilde{X_4}[k]+\tilde{Z_3}[k],\tilde{h_{45}}[k]\tilde{X_4}[k]+\tilde{Z_5}[k], \cdots \\
&\mathbf{\tilde{h_{4,2K-1}}}[k]\tilde{X_4}[k]+\mathbf{\tilde{Z_{2K-1}}}[k],\tilde{Z_{2K}}[k]|\tilde{h_{41}}\tilde{X_4}[k]+\tilde{Z_1}[k])\\
&\;\;\;\; -H(\tilde{Z}_{(2,...,2K)/4}[k])]\\
&\leq n(\log P +o(\log P))+no(\log P),
\end{align*}
Then, 
\begin{align*}
d_{12}+d_{34} + d_{21}+d_{43} &\leq \limsup_{P\rightarrow\infty} \frac{R_{12}+R_{34} +R_{21}+R_{43}}{\log (P)}\\
&  \leq 1+0+0 + 1+0+0 = 2,
\end{align*}
Summing over all rate pairs (see Remark 5) leads to the theorem, which indicates that the causal MIMO relay cannot increase the DoF of the full-duplex two-pair user two-way IC. 

\end{proof}
\begin{remark}
We are able to sum over all rate pairs because the asymmetry of the transformed network (multiple antennas at user $2K-1$ only, and user $2K-1$ is connected to all other nodes, unlike the even and odd numbered nodes)  does not affect the DoF. Intuitively this is because for a SIMO or MISO point-to-point channel, the DoF is still 1. More rigorously,  consider the sum rate pair $R_{12}+R_{2K-1,2K}$ and using similar notation (now $M_A$ denotes all messages except $M_{12},M_{2K-1,2K}$), and following the same steps as in bounding $R_{12}+R_{34}$, we notice that the bounds do not depend on the asymmetry and again lead to 1 DoF per pair:\footnote{We leave out several steps and replace it with $\cdots$ to avoid repetition, as these follow identically.}
\begin{align*}
&n(R_{12}+R_{2K-1,2K}-\epsilon) \\
&\leq I(M_{2K-1,2K};\tilde{Y_{2K}}^n|M_A)\\
& \;\;+I(M_{12};\tilde{Y_{2K}}^n,\tilde{Y}^n_{(2,...,2K-1)}|M_{2K-1,2K},M_A)\\
&  \leq ...\\
&\leq \sum_{k=1}^n [H(\tilde{Y_{2K}}[k])-H(\tilde{Z_{2K}}[k])\\
&+\sum_{k=1}^n[H(\tilde{h_{12}}[k]\tilde{X_1}[k]+\tilde{Z_2}[k], \tilde{Z_3}[k],\tilde{h_{14}}[k]\tilde{X_1}[k]+\tilde{Z_4}[k], \tilde{Z_5}[k],...,\\
&\mathbf{\tilde{h_{1,2K-1}}}[k]\tilde{X_1}[k]+\mathbf{\tilde{Z_{2K-1}}}[k]|\tilde{h_{1,2K}}\tilde{X_1}[k]+\tilde{Z_{2K}}[k])\\
&\;\;\;\; -H(\tilde{Z}_{(2,...,2K-1)}[k])]\\
&\leq n(\log P +o(\log P))+no(\log P),
\end{align*}
Thus we will have $d_{12}+d_{2K-1,2K}\leq 1$. Similar arguments follow for the opposite direction.
\end{remark}


\section{conclusion}
\label{conclusion}
We proposed and studied  the $K$-pair-user two-way interference channel with and without a MIMO relay where all nodes operate in full duplex. 
We demonstrated that the degrees of freedom of the $K$-pair-user two-way IC without a relay is $K$, which indicates that full-duplex operation doubles the DoF over the setting with half-duplex nodes for this two-way setting, but that interaction, or adapting transmission based on previously received signals at the users cannot further increase the DoF beyond what full-duplex allows, i.e. the DoF is just that of two one-way, non-interactive ICs. 
We next showed that if we introduce a $2K$ antenna, full-duplex and non-causal relay, that the DoF may again be doubled over the full-duplex, relay-free counterpart (or quadrupled over the half-duplex counterpart). We demonstrated a one-shot scheme to achieve the maximal $2K$ DoF. In sharp contrast, if the relay is causal rather than non-causal, we derived a new converse showing that the DoF cannot be increased beyond $K$ for a $K$-pair-user two-way full-duplex IC.  
We commented on how one may decrease the number of antennas at the relay node, at the expense of either a reduced achievable DoF or cognition at the relay. However, a converse for the $K$-pair user TWIC with an instantaneous relay with fewer than $2K$ antennas is open. 
Overall, this work has shown that in $K$-pair-user two-way interference channels, full-duplex operation at least doubles the achievable DoF (over half-duplex systems), interaction does not help (unless some channel gains are zero), and a full-duplex relay may further increase the DoF (quadrupling the DoF over a half-duplex system) if it is instantaneous and has a sufficient number of antennas.

\section*{Acknowledgement}
The authors would like to thank Tang Liu, a Ph.D. student at UIC, for suggesting the ``one-shot'' scheme in Section \ref{2KDoF}, which is simpler than the authors' previously derived block Markov coding scheme, to achieve $2K$ DoF.

\bibliographystyle{IEEEtran}
\bibliography{refs}

\begin{thebibliography}{10}
\providecommand{\url}[1]{#1}
\csname url@samestyle\endcsname
\providecommand{\newblock}{\relax}
\providecommand{\bibinfo}[2]{#2}
\providecommand{\BIBentrySTDinterwordspacing}{\spaceskip=0pt\relax}
\providecommand{\BIBentryALTinterwordstretchfactor}{4}
\providecommand{\BIBentryALTinterwordspacing}{\spaceskip=\fontdimen2\font plus
\BIBentryALTinterwordstretchfactor\fontdimen3\font minus
  \fontdimen4\font\relax}
\providecommand{\BIBforeignlanguage}[2]{{%
\expandafter\ifx\csname l@#1\endcsname\relax
\typeout{** WARNING: IEEEtran.bst: No hyphenation pattern has been}%
\typeout{** loaded for the language `#1'. Using the pattern for}%
\typeout{** the default language instead.}%
\else
\language=\csname l@#1\endcsname
\fi
#2}}
\providecommand{\BIBdecl}{\relax}
\BIBdecl

\bibitem{Katti:2011}
M.~Jain, J.~Choi, T.~Kim, D.~Bharadia, S.~Seth, K.~Srinivasan, P.~Levis,
  S.~Katti, and P.~Sinha, ``Practical, real-time, full duplex wireless,'' in
  \emph{Mobicom}, Las Vegas, NV, Sep. 2011.

\bibitem{Sahai:2012}
A.~Sahai, G.~Patel, and A.~Sabharwal, ``Asynchronous full-duplex wireless,'' in
  \emph{International Conference on Communication Systems and Networks
  (COMSNETS)}, Jan. 2012.

\bibitem{Han:1984}
T.~Han, ``A general coding scheme for the two-way channel,'' \emph{IEEE Trans.
  Inf. Theory}, vol. IT-30, pp. 35--44, Jan. 1984.

\bibitem{zcheng_ISIT}
Z.~Cheng and N.~Devroye, ``On the capacity of multi-user two-way linear
  deterministic channels,'' in \emph{Proc. IEEE Int. Symp. Inf. Theory},
  Cambridge, Jul. 2012.

\bibitem{Cheng:2012}
\BIBentryALTinterwordspacing
------, ``Two-way networks: when adaptation is useless,'' submitted to IEEE
  Trans. on Inf. Theory, June 2012. [Online]. Available:
  \url{http://arxiv.org/abs/1206.6145}
\BIBentrySTDinterwordspacing

\bibitem{zcheng_Allerton2012}
------, ``On constant gaps for the two-way gaussian interference channel,'' in
  \emph{Proc. Allerton Conf. Commun., Control and Comp.}, Oct. 2012.

\bibitem{Suh2012}
C.~Suh, I.-H. Wang, and D.~Tse, ``Two-way interference channels,'' in
  \emph{Proc. IEEE Int. Symp. Inf. Theory}, Cambridge, Jul. 2012.

\bibitem{etkin_tse_wang}
R.~Etkin, D.~Tse, and H.~Wang, ``Gaussian interference channel capacity to
  within one bit,'' \emph{IEEE Trans. Inf. Theory}, vol.~54, no.~12, pp.
  5534--5562, Dec. 2008.

\bibitem{Cheng:GlobalSIP}
Z.~Cheng and N.~Devroye, ``On constant gaps for the {K}-pair user two-way
  {G}aussian interference channel with interaction,'' in \emph{Proc. of IEEE
  Global Conf. on Signal and Information Processing (GlobalSIP)}, Austin, Dec.
  2013.

\bibitem{Zheng:Dof}
L.~Zheng and D.~Tse, ``Communication on the grassmann manifold: a geometric
  approach to the noncoherent multiple-antenna channel,'' \emph{IEEE Trans.
  Inf. Theory}, vol.~48, no.~2, pp. 359--383, 2002.

\bibitem{Jafar:2008:alignment}
V.~Cadambe and S.~Jafar, ``Interference alignment and the degrees of freedom
  for the {K}-user interference channel,'' \emph{IEEE Trans. Inf. Theory},
  vol.~54, no.~8, pp. 3425--3441, Aug. 2008.

\bibitem{Jafar:2009:relays}
------, ``Degrees of freedom of wireless networks with relays, feedback,
  cooperation and full duplex operation,'' \emph{IEEE Trans. Inf. Theory},
  vol.~55, no.~5, pp. 2334--2344, May 2009.

\bibitem{Jafar:2009:X}
S.~A. Jafar and S.~Shamai, ``Degrees of freedom region of the {M}{I}{M}{O} {X}
  channel,'' \emph{IEEE Trans. Inf. Theory}, vol.~54, no.~1, pp. 151--170, Jan.
  2008.

\bibitem{Etkin:dof}
R.~Etkin and E.~Ordentlich, ``The degrees-of-freedom of the k-user gaussian
  interference channel is discontinuous at rational channel coefficients,''
  \emph{IEEE Trans. Inf. Theory}, vol.~55, no.~11, pp. 4932--4946, Nov. 2009.

\bibitem{Ke2011}
L.~Ke, A.~Ramamoorthy, Z.~Wang, and H.~Yin, ``Degrees of freedom region for an
  interference network with general message demands,'' in \emph{Proc. IEEE Int.
  Symp. Inf. Theory}, 2011, pp. 36--40.

\bibitem{Vaze:dof}
C.~Vaze and M.~Varanasi, ``The degrees of freedom region of the two-user mimo
  broadcast channel with delayed csit,'' in \emph{Proc. IEEE Int. Symp. Inf.
  Theory}, 2011, pp. 199--203.

\bibitem{Lee:2013}
\BIBentryALTinterwordspacing
N.~Lee and R.~Heath, ``Multi-way information exchange over completely-connected
  interference networks with a multi-antenna relay,'' 2013. [Online].
  Available: \url{http://arxiv.org/abs/1302.0749}
\BIBentrySTDinterwordspacing

\bibitem{Wang:2013}
\BIBentryALTinterwordspacing
C.~Wang and S.~Jafar, ``Degrees of freedom of the two-way relay mimo
  interference channel,'' 2013. [Online]. Available:
  \url{http://escholarship.org/uc/item/9qc3343h}
\BIBentrySTDinterwordspacing

\bibitem{maddah-ali:IA}
M.~Maddah-Ali, A.~Motahari, and A.~Khandani, ``Communication over x channel:
  Signaling and performance analysis,'' in \emph{Univ. of Waterloo, Waterloo,
  ON, Canada, Tech. Rep. UW-ECE-2006-27}, Dec. 2006.

\bibitem{Jafar:IA}
S.~Jafar, ``Interference alignment: A new look at signal dimensions in a
  communication network,'' in \emph{Foundations and Trends in Communications
  and Information Theory}, no.~1, 2011, pp. 1--134.

\bibitem{Motahari2009}
A.~Motahari, S.~Oveis-Gharan, M.~Maddah-Ali, and A.~Khandani, ``Real
  interference alignment: Exploiting the potential of single antenna systems,''
  in \emph{arXiv:0908.2282v2}, Nov. 2009.

\bibitem{Jafar_Gdof2010}
S.~Jafar and S.~Vishwanath, ``Generalized degrees of freedom of the symmetric
  gaussian user interference channel,'' \emph{IEEE Trans. Inf. Theory},
  vol.~56, no.~7, pp. 3297--3303, Jul. 2010.

\bibitem{Mohajer2012}
S.~Mohajer, R.~Tandon, and H.~Poor, ``Generalized degrees of freedom of the
  symmetric k-user interference channel with feedback,'' in \emph{Proc. IEEE
  Int. Symp. Inf. Theory}, 2012, pp. 3125--3129.

\bibitem{Shomorony_Allerton}
I.~Shomorony and S.~Avestimehr, ``Degrees of freedom of two-hop wireless
  networks: {E}veryone gets the entire cake,'' in \emph{Proc. Allerton Conf.
  Commun., Control and Comp.}, Oct. 2012.

\bibitem{ElGamal2005}
A.~El~Gamal and N.~Hassanpour, ``Relay-without-delay,'' in \emph{Information
  Theory, 2005. ISIT 2005. Proceedings. International Symposium on}, 2005, pp.
  1078--1080.

\bibitem{LeeWang2013}
N.~Lee and C.~Wang, ``Aligned interference neutralization and the degrees of
  freedom of the two-user wireless networks with an instantaneous relay,''
  \emph{IEEE Trans on Comm.}, vol.~PP, no.~99, pp. 1--9, 2013.

\bibitem{costa_interference}
M.~Costa, ``On the gaussian interference channel,'' \emph{IEEE Trans. Inf.
  Theory}, vol.~31, no.~5, pp. 607--615, Sep. 1985.

\bibitem{sason}
I.~Sason, ``On achievable rate regions for the {G}aussian interference
  channel,'' \emph{IEEE Trans. Inf. Theory}, Jun. 2004.

\bibitem{Sato:degraded}
H.~Sato, ``On degraded gaussian two-user channels,'' \emph{IEEE Trans. Inf.
  Theory}, vol.~24, no.~5, pp. 638--640, 1978.

\bibitem{Goldsmith:2005}
A.~Goldsmith, \emph{Wireless Communications}, 2nd~ed.\hskip 1em plus 0.5em
  minus 0.4em\relax Cambridge University Press, 2005.

\bibitem{devroye_IEEE}
N.~Devroye, P.~Mitran, and V.~Tarokh, ``Achievable rates in cognitive radio
  channels,'' \emph{IEEE Trans. Inf. Theory}, vol.~52, no.~5, pp. 1813--1827,
  May 2006.

\bibitem{host_madsen_int}
A.~Host-Madsen, ``Capacity bounds for cooperative diversity,'' \emph{IEEE
  Trans. Inf. Theory}, vol.~52, pp. 1522--1544, Apr. 2006.

\end{thebibliography}

\end{document}